\documentclass[onecolumn,draftclsnofoot, 12pt, ]{IEEEtran}
\usepackage{times}
\usepackage[T1]{fontenc}
\usepackage{graphicx}
\usepackage{amsmath}
\usepackage{setspace}
\usepackage{amsbsy}
\usepackage{amsthm}
\usepackage{cite}
\usepackage{cite}
\usepackage{authblk}
\usepackage{subfigure}
\usepackage{amsfonts}
\usepackage{array}
\usepackage{epstopdf}
\usepackage{algorithm}
\usepackage{algorithmicx}
\usepackage{algpseudocode}
\newtheorem{theorem}{Theorem}

\newtheorem{corollary}[theorem]{Corollary}

\begin{document}
\title{Zero-Forcing Per-Group Precoding (ZF-PGP)
 for
Robust Optimized Downlink Massive MIMO Performance
}
\author{Thomas~Ketseoglou,~\IEEEmembership{Senior~Member,~IEEE,} and Ender~Ayanoglu,~\IEEEmembership{Fellow,~IEEE}

\thanks{T. Ketseoglou is with the Electrical and Computer Engineering Department, California State Polytechnic University, Pomona, California (e-mail: tketseoglou@cpp.edu). E. Ayanoglu is
with the Center for Pervasive Communications and
Computing, Department of Electrical Engineering and Computer Science,
University of California, Irvine (e-mail: ayanoglu@uci.edu). This work was partially supported by NSF grant 1547155.}}

\maketitle
\begin{abstract}
In this paper, we propose a new, combined Zero-Forcing Per-Group Precoding (ZF-PGP) method that achieves very high gains in comparison to ZF Precoding techniques, while it simultaneously offers individually separate streams to reach individual User Equipment (UE), i.e., it obliterates the need for coordinated, joint decoding by the group's UEs. Although ZF-PGP in general experiences a performance loss in comparison to the near-optimal Per-Group Precoding within groups (PGP-WG), it is much simpler to implement than PGP-WG. The method works extremely well in conjunction with Combined Frequency and Space Division and Multiplexing (CFSDM) in correlated channels or in cases in which the group's total number of antennas is greater than the group's available Virtual Channel Model Beams (VCMB) number. As we show by multiple examples, for a Uniform Linear Array (ULA), by combining the proposed technique with massive Multiple-Input Multiple-Output (MIMO) and Joint Spatial Division and Multiplexing for Finite Alphabets (JSDM-FA) type hybrid precoding or CFSDM a very efficient, near-optimal design with reduced complexity at the UEs is achieved, while each UE decodes its received data independently, without interference from data destined for other UEs. Due to JSDM-FA Statistical Channel State Information (SCSI) inner beamforming, the groups are easily identified. We show that the combined ZF-PGP scheme achieves very high gains in comparison to ZF and Regularized ZF (RZF) precoding and other modified sub-optimal techniques investigated, while its complexity is still relatively low due to PGP-WG. We also demonstrate the robustness of the proposed precoding techniques to channel estimation errors, showing minimal performance loss.
\end{abstract}
\IEEEpeerreviewmaketitle
\section{{Introduction}}
Massive MIMO  employs a large number of antennas in order to achieve very high spectral efficiency \cite{Marzetta1,Marzetta2,Marzetta3}. Due to extremely high data rate requirements in massive MIMO, downlink precoding is essential toward increased downlink data rates. Thus, for massive MIMO to be capable of offering its full benefits, quite accurate Channel State Information (CSI) is required at the base station (BS) for efficient downlink precoding purposes. There are techniques proposed for downlink linear precoding in a multi-user MIMO scenario, e.g., Joint Spatial Division and Multiplexing (JSDM) \cite{JSDM1,JSDM2, JSDM3}.
JSDM divides users in groups that are orthogonal and applies ZF or RZF precoding for the users of each group which offers low decoding complexity. However, the spectral efficiency of JSDM is low. A simplification to user grouping for JSDM with Finite Alphabet inputs (JSDM-FA) is proposed in \cite{TK_EA_TCOM2}, where the common sparse support of geographically similar users is employed to implement JSDM and a PGP-WG \cite{TK_EA_QAM} near-optimal linear precoder is employed for each group to offer outer precoding with high performance and with low complexity.  In addition, the JSDM-FA approach resorts to Statistical Channel State Information  (SCSI) to derive the user groups, i.e., its inner precoder is slowly varying with time. However, \cite{TK_EA_TCOM2} shows that without CFSDM, the users in a group need to jointly demodulate their information in order to achieve the benefits of PGP-WG. On the other hand, \cite{TK_EA_TCOM2} shows that with CFSDM and without joint user decoding there is a loss of about 25\% in spectral efficiency. Therefore, the presented schemes, although achieving near optimal linear precoding performance, fall short in the user decoding process, due to high complexity.

 In other past related work, \cite{Max} has presented an iterative algorithm for precoder optimization for sum rate maximization of  Multiple Access Channels (MAC) with Kronecker MIMO channels. Furthermore, more recent work has shown that when only SCSI\footnote{SCSI pertains to the case in which the transmitter has knowledge of only the MIMO channel correlation matrices, or in general slow-varying parameters \cite{Khan,Weich} and the thermal noise variance.} is available at the transmitter, in asymptotic conditions when the number of transmitting and receiving antennas grows large, but with a constant transmitting to receiving antenna number ratio, one can design the optimal precoder by looking at an equivalent constant channel and its corresponding adjustments as per the pertinent theory \cite{SCSI}, and applying a modified expression for the corresponding ergodic mutual information evaluation over all channel realizations. This development allows for a precoder optimization under SCSI in a much easier way \cite{SCSI}. Finally, \cite{LozanoA,TK_EA_QAM} presented for the first time results for mutual information maximizing linear precoding with large size MIMO configurations and QAM constellations. Such systems are particularly difficult to analyze and design when the inputs are from a finite alphabet, especially with QAM constellation sizes of $M\geq 16$. On the other hand, for non-linear precoding, recently \cite{Schoeber_THP,Schoeber_THP_TWC} has presented a robust, two-stage precoding approach, with channel estimation errors and for a hierarchical  two-stage (hybrid) precoding design and with an outer Tomlinson-Harashima (TH) precoder per group that aims at minimizing the maximum mean square error (MSE) per user for each group. These non-linear outer precoders achieve near Dirty Paper Coding (DPC) results \cite{DPC}, i.e., they almost attain the channel capacity of a Gaussian Broadcast Channel (GBC)\footnote{This GBC capacity (system throughput)\cite{DPC} precludes receiver cooperation.}, due to the two-stage JSDM-type approach.

 However, near-optimal mutual information linear precoding approach \cite{Xiao, TK_EA_TCOM2, LozanoA} offers the possibility of achieving the highest possible spectral efficiency under perfect or estimated CSI, in many cases significantly higher than the (non-precoded) group channel mutual information (GCMI) with full cooperation at the receiver\footnote{We use the term GCMI to mean the input-output mutual information with equally likely input symbols and without any type of precoding applying to the input of the MIMO channel.}. The latter is a strict upper bound for many techniques, including ZF, RZF, and TH non-linear precoding \cite{Schoeber_THP,DPC}. On the other hand, as \cite{TK_EA_TCOM2} has shown, the linear precoding approach requires cooperation at the demodulation of each group and perfect CSI. It is thus still an open and important problem to address the issue of near-optimal downlink precoding in conjunction with user-independent decoding, i.e., without Multiple Access Interference (MAI) present at the UEs and under mismatched decoding at the UE.

In this paper, we present near-optimal linear precoding techniques for downlink massive MIMO, suitable for QAM with constellation size $M\geq 16$ and CSIT at the BS with either a Uniform Linear Array (ULA) or a Uniform Planar Array (UPA) configuration, without the requirement for co-ordinated decoding at the UE. We show that by combining a ZF precoder (ZFP) with a PGP-WG one, a major simplification in the user separation within a group (USWG) problem results, i.e., simple decoding for each user. At the same time, the system can achieve the same near-optimal performance of PGP-WG achievable in the original downlink channel (without the ZF part) in high signal-to-noise ratio $\mathrm{SNR}$\footnote{In lower $\mathrm{SNR}$ the performance is still much higher than the one achieved by ZF-type precoding techniques.}. Thus, the same spectral efficiency results, as described in \cite{TK_EA_TCOM2}, can be achieved in high $\mathrm{SNR}$, albeit with much simpler decoding at the UE, as separate independent data streams are delivered to users. This high performance remains valid, even with channel estimation errors, i.e., the ZF-PGP precoder is robust. This has important and essential ramifications to linear downlink precoding for massive MIMO because an alternative technique employed in \cite{TK_EA_TCOM2} based on deploying orthogonal subcarriers to different users (CFSDM) results in significantly lower spectral efficiency. Finally, as the presented numerical results show, ZF-PGP performs dramatically better than ZF, RZF, and other suboptimal techniques discussed in this paper. Thus, ZF-PGP has many advantages over other linear precoding techniques.

The contributions of this paper can be summarized as follows:
\begin{enumerate}
\item It presents an analytical framework that builds on the concept of JSDM-FA and allows for near-optimal linear precoding on the downlink while it achieves fully separated data streams for each user by combining ZF with PGP-WG precoding.
    \item It shows that the ZF-PGP precoder results in an interesting effective channel representation of the overall downlink transmission chain.
\item It shows that the presented approach can achieve the full benefit of PGP-WG in high $\mathrm{SNR}$, i.e., it achieves the same power efficiency, which was shown to be near-optimal in the past.
\item It studies the problem of errors between the PGP-WG ideal (perfect CSI) optimal precoder employed for each group at the BS and the one resulting from estimating the channel with errors, i.e., with imperfect CSI, and shows only marginal loss for power efficiency.
\item For the type of channels that are present in massive MIMO, it shows that ZF-PGP significantly outperforms other commonly employed linear precoding techniques such as ZF, RZF, and other suboptimal precoding techniques even with imperfect CSI.
\end{enumerate}

The paper is organized as follows: Section II presents the system model and problem statement together with a short introduction of the JSDM-FA efficient downlink linear precoding of  \cite{TK_EA_TCOM2} in a JSDM fashion for ULA narrowband channels and the PGP-WG method. Section III presents the ZF-PGP precoder concept that allows for optimized mutual information precoding in each group, while it distributes the information to each user independently, i.e., without any cross-user interference. Then, the ZF-PGP concept is generalized to cases where the ZF-PGP precoder is estimated at the UE , i.e., under a downlink channel estimation with errors model. In Section IV, numerical results are presented for ZF-PGP with and without CFSDM and for ideal CSI as well as CSI with errors. Comparisons with many other precoding techniques are also made. Finally, our conclusions are presented in Section V.

\underline{Notation:} We use small bold letters for vectors and capital bold letters for matrices. ${\mathbf A}^T$, ${\mathbf A}^H$, ${\mathbf A}^*$, ${\mathbf A}_{\cdot, i}$, and ${\mathbf A}_{i,\cdot}$ denote the transpose, Hermitian conjugate, complex conjugate, column $i$, and row ${i}$ of matrix ${\mathbf A}$, respectively. ${\mathbf S}^T$ denotes a selection matrix, i.e., of size $k\times n$ with $k<n$ consists of rows equal to different unit row vectors ${\mathbf e}_i$ where the row vector element $i$ is equal to $1$ in the $i$th position and is equal to $0$ in all other positions, the specific ${\mathbf e}_i$ vectors used are defined by the desired selection.  ${\mathbf F}_{N}$ denotes the DFT matrix of order $N$. We use ${\mathbf h}_{u,g,k,n}$ for the uplink channel of user $k$'s antenna $n$ in group $g$. ${\mathbf H}_g$ is the uplink channel of group $g$, while ${\tilde {\mathbf H}}_g$ is its projection to the Virtual Channel Model (VCM) basis. The overall virtual channel, reduced size representation representation of ${\mathbf H}_g$ is denoted by ${\mathbf H}_{g,v}$. ${\mathbf H}_{d,g}$ represents the overall downlink channel for group $g$, i.e., due to time division duplex reciprocity ${\mathbf H}_{d,g} = {\mathbf H}_g^H$. For a ULA or a UPA, ${\mathbf H}_{u,g,k,n}$ represents the uplink channel of user $k$'s antenna $n$ in group $g$, then ${\tilde {\mathbf H}}_{u,g,k,n}$ is its projection to the two DFT matrices (VCM bases), ${\tilde{\mathbf h}}_{u,g,k,n}$ is its corresponding vectorized form.

\section{{System Model and Problem Statement}}
Consider the downlink precoding equation on a narrowband (flat-fading) massive MIMO system with a single cell and JSDM \cite{JSDM1}
\begin{equation}
{\mathbf y_d} = {\mathbf H}_u^H {\mathbf P} {\mathbf x}_d + {\mathbf n_d}, \label{original}
\end{equation}
where ${\mathbf y}_d$ is the downlink received vector of size $\sum_{g=1}^G N_{d,g}\times 1$, ${\mathbf x}_d$ is the $N_u\times 1$ vector of transmitted symbols drawn independently from a QAM constellation, where the downlink channel matrix ${\mathbf H}_d={\mathbf H}_u^H$, where ${\mathbf H}_u = [{\mathbf H}_1,\cdots, {\mathbf H}_G]$ is the $N_u\times K_{eff}$ uplink channel matrix from all $K$ users, employing $N_u$ receiving antennas at the BS, with $K_{eff}= \sum_{g} N_{d,g}$, where $N_{d,g}$ is the total number of antennas of all users in group $g$. Users have been divided into $G$ groups with $K_g$ users in group $g$ ($1\leq g \leq G$), with user $k$ of group $g$ denoted as $k^{(g)}$. In this paper, we assume that each UE employs (without loss of generality) $N_{d,k^{(g)}}=2$  antennas, with ($\sum_{g=1}^G K_g = K$), ${\mathbf H}_g =[{\mathbf H}_{g^{(1)}}\cdots {\mathbf H}_{g^{(K_g)}} ]$ being group $g$'s uplink channel matrix of size $N_u\times N_{d,g}$, with $N_{d,g}$ comprising the total number of antennas in the group, and where ${\mathbf n}_d$ represents the independent, identically distributed (i.i.d.) complex circularly symmetric Gaussian noise of mean zero and variance per component $\sigma_u^2 = \frac{1}{\mathrm{SNR}_{s,d}}$, where ${\mathrm{SNR}_{s,d}}$ is the channel symbol $\mathrm {SNR}$ on the downlink. The downlink symbol vector ${\mathbf x}_g$ of size $\sum_{g}^G{N_{d,g}}\times 1$ has i.i.d. components drawn from a QAM constellation of order $M$. We assume that Time Division Duplexing (TDD) is employed in the system, to be able to exploit the reciprocity between the uplink and downlink channels.
The optimal linear precoder ${\mathbf P}$ needs to satisfy
\begin{eqnarray}
\begin{aligned}
& \underset{\mathbf P}{\text{maximize}}
& & I({\mathbf x_d};{\mathbf y_d}|\hat{\mathbf H}_u)\\
& \text{subject to}
& &  \mathrm {tr}({\mathbf P} {\mathbf P}^H) = \sum_{g=1}^G N_{d,g}, \\ \label{eq_MMIMO}
\end{aligned}
\end{eqnarray}
where the conditioning indicates the estimated channel CSI and the constraint is due to keeping the total power emitted to the totality of downlink antennas with precoding equal to the one without precoding.

In order to reduce the  complexity involved in (\ref{eq_MMIMO}), JSDM was proposed in \cite{JSDM1}. JSDM divides users into approximately orthogonal groups assuming a Gaussian channel \cite{JSDM1}, based on approximately equal channel covariance matrices in each group. Because of the resulting approximate orthogonality between different groups, $I({\mathbf x_d};{\mathbf y_d})= \sum_{g=1}^G I({\mathbf x}_g; {\mathbf y}_g)$, where ${\mathbf x}_g ~\text{and}~ {\mathbf y}_g$ represent the data symbols and received data of group $g$, respectively (see (11) and (12) below and \cite{JSDM1}). Thus, the problem in (\ref{eq_MMIMO}) becomes equivalent to the one that maximizes the sum of the group information rates. The downlink precoder is then divided into two parts, i.e., it becomes a hybrid precoder comprising two stages: a) An inner pre-beamforming stage, and b) An outer Multi-User MIMO (MU-MIMO) precoder stage. On the other hand, JSDM-FA was proposed in \cite{TK_EA_TCOM2} to facilitate group determination and simultaneously reduce the dimensionality of the outer (intra-group) precoding problem, based on a projection of the actual uplink channels to the DFT orthogonal base of size $N_u$. We shortly review the JSDM-FA concept here, then we employ it for the purpose of this paper.

Assume without loss in generality that a ULA deployed at the BS along the $z$ direction and for flat fading, i.e., $B< B_{COH}$, where $B~\text{and}~B_{COH}$ are the RF signal bandwidth and the coherence bandwidth of the channel, respectively. Each user group on the uplink transmits from the same ``cluster'' of elevation angles $\theta_{g} \in [{\bar \theta}_{g} -\Delta \theta, {\bar \theta}_{g} +\Delta \theta]$, with ${\bar \theta}_{g}$ being the mean of ${\theta}_{g}$, distributed uniformly in the support interval, thus each user's $k^{(g)}$ of group $g$, ($1\leq k^{(g)} \leq K_g$ and $1\leq g \leq G$) transmitting antenna $n$ channel, ${\mathbf h}_{u,g,k,n} = \frac{1}{{\sqrt L}}\sum_{l=1}^L \beta_{lgkn} {\mathbf a}(\theta_{lgkn})$, where ${\mathbf a}(\theta_{lgkn}) = [1, \exp(-j{2\pi}D\cos(\theta_{lgkn})),\cdots, \exp(-j{2\pi D(N_u-1)}\cos(\theta_{lgkn}))]^T$ is the array response vector, where each $\theta_{lgkn}$ is independently selected and uniformly distributed in the group's angular support $[{\bar \theta}_{g} -\Delta \theta, {\bar \theta}_{g} +\Delta \theta]$ of its group, with $D=d/\lambda$ representing the normalized distance of successive array elements, $\lambda$ is the wavelength, $\theta_{lgkn}$ is the elevation (arrival) angle of the  $l$ path of group $g$ $k$ user's $n$ receiving antenna, and the path gains $\beta_{lgkn}$ are independent complex Gaussian random variables with zero mean and variance $1$. This channel model is similar to the one in \cite{Filippou}. Note that the channel model adopted here is not Gaussian if $L$ is relatively small \cite{Filippou}. The VCM representation, presented in \cite{Sayeed}, is formed by projecting the original channel ${\mathbf H_u}$ to the $N_u$ dimensional space formed by the $N_u\times N_u$ DFT matrix ${\mathbf F}_{N_u}$. For massive MIMO systems, i.e., when $N_u \gg 1$,
by projecting each group channel ${\mathbf H}_{g}$ on the DFT virtual channel space \cite{Sayeed}, we get
${\tilde {\mathbf H}}_{g,v} = {\mathbf F}_{N_u}^H {\mathbf H}_g$,
where
${\mathbf F}_{N_u}$ is the DFT matrix of order $N_u$. Since each group attains the same angular behavior, over all users and antennas in the group, only a few, consecutive elements of ${\tilde {\mathbf H}}_g$ are significant \cite{TK_EA_TCOM2}. This comes as a result of the fact that significant angular components need to be in the  main lobe of the response vector (see \cite{TK_EA_TCOM2} for details).
By employing a size $|{\cal S}_g|\times N_u$ selection matrix ${\mathbf S}_g$\footnote{A selection matrix ${\mathbf S}^T$ of size $k\times n$ with $k<n$ consists of rows equal to different unit row vectors ${\mathbf e}_i$ where the row vector element $i$ is equal to  $1$ in the $i$th position and is equal to $0$ in all other positions, with the specific ${\mathbf e}_i$ vectors employed defined by the desired selection. Such a matrix has the property that ${\mathbf S}^T {\mathbf S} = {\mathbf I}$.}
\begin{equation}
{{\mathbf H}_{g,v}}= {\mathbf S}_g^T{\tilde{\mathbf H}}_g={\mathbf S}_g^T{\mathbf F}_{N_u}^H {\mathbf H}_g, \label{virtual_ch}
\end{equation}
where the group $g$ virtual channel matrix ${{\mathbf H}_{g,v}}$ is a reduced size, $r_g\times N_{d,g}$, matrix, with $r_g = |{\cal S}_{g}|$ being the number of significant angular components in group $g$, due to the sparsity available in the angular domain. We can then write for the uplink group $g$ channel matrix ${\mathbf H}_g$,
\begin{equation}
{\mathbf H}_g = {\mathbf F}_{N_u}{\mathbf S}_g{\mathbf S}_g^T{\mathbf F}_{N_u}^H {\mathbf H}_g={\mathbf F}_{N_u, {\cal S}_g}{{\mathbf H}_{g,v}},
\end{equation}
where ${\mathbf F}_{N_u, {\cal S}_g}$ represents the selected columns of ${\mathbf F}_{N_u}$ due to its sparse representation in the angular domain.
Finally, due to non-overlapping of the support sets, i.e., ${\cal S}_n\cap_{m\neq n}{\cal S}_m =\emptyset$, we see that the system becomes approximately orthogonal inter-group wise, i.e., $\sum_{m\neq g}{\mathbf H}_{d,g }{\mathbf H}_{d,m }^H \approx 0$. Then,
\begin{equation}
\begin{split}
 {\mathbf y}_d =  \left[ \begin{array}{c} {\mathbf H}_{1,v}^H{\mathbf F}_{N_u, {\cal S}_1}^H \\
{\mathbf H}_{2,v}^H{\mathbf F}_{N_u, {\cal S}_2}^H \\
 \vdots\\
 {\mathbf H}_{G,v}^H{\mathbf F}_{N_u, {\cal S}_G}^H \end{array} \right]
 \left[ \begin{array}{c c c c} {\mathbf F}_{N_u, {\cal S}_1} &
{\mathbf F}_{N_u, {\cal S}_2} & \cdots &
{\mathbf F}_{N_u, {\cal S}_G} \end{array} \right]\\
\left[
\begin{array}{cccccc}
{\bf P}_1 & {\bf 0} & {\bf 0} & \cdots & {\bf 0} & {\bf 0}\\
{\bf 0} & {\bf P}_2 & {\bf 0} & \cdots & {\bf 0} & {\bf 0}\\
{\bf 0} & {\bf 0} & {\bf P}_3 & \cdots & {\bf 0} & {\bf 0}\\
\vdots & \vdots & \vdots & \ddots & \vdots & \vdots\\
{\bf 0} & {\bf 0} & {\bf 0} & \cdots &{\bf P}_{G-1} & {\bf 0}\\
{\bf 0} & {\bf 0} & {\bf 0} & \cdots &{\bf 0} & {\bf P}_G\\
\end{array}
\right]
 \left[ \begin{array}{c} {{\mathbf x}_{1}} \\
{{\mathbf x}_{2}} \\
 \vdots\\
 {{\mathbf x}_{G}} \end{array} \right]+ {\mathbf n}, \label{full_jsdm}
 \end{split}
\end{equation}
where for $1\leq g \leq G$, ${\mathbf H}_{g,v}^H$  is a size $N_{d,g}\times |{\cal S}_g|$ matrix, ${\mathbf F}_{N_u, {\cal S}_G}$ is a size $|{\cal S}_g| \times N_u$ matrix, ${\mathbf P}_{g}$ is a size $|{\cal S}_g|\times |{\cal S}_g|$ matrix, and ${\mathbf x}_g$ is the group $g$ downlink symbol vector of size $|{\cal S}_g|\times 1$. As shown in \cite{TK_EA_TCOM2}, the performance of the system is very good, however the MU-MIMO PGP-WG precoder needs cooperation among all receivers in the group in order to achieve this near-optimal performance. More explicitly, PGP-WG requires that each user in a JSDM-FA group has knowledge of the received data of other users in the group in order to decode his data. An alternative we considered in \cite{TK_EA_TCOM2} is to employ multiple carriers for each user, i.e., Combined Frequency and Space Division and Multiplexing (CFSDM), but at a significant cost of power efficiency. Another alternative way is to investigate how the right singular vector matrix on the downlink channel of a group can be eliminated in the transmission phase, i.e., an integration of a ZFP into the PGP-WG framework needs to be studied. In the sequel, we show that by employing a combined ZF-PGP approach, we can achieve the same power efficiency as the original PGP-WG of JSDM-FA per group at high $\mathrm{SNR}$, albeit without requiring any receiver cooperation, i.e., with much smaller complexity. The rationale for this is the fact that ZF precoding creates independent data streams to each receiving antenna. Furthermore, in \cite{TK_EA_TCOM2} it is shown that precoding using the VCM model representation is equivalent to precoding in the original model of (\ref{original}). We list here two important properties associated with JSDM-FA:
\begin{enumerate}
\item It is based on SCSI, i.e., on the elevation angles of the uplink channel, which are fixed over long periods of time, and
\item It offers an information lossless lower dimension representation of the original channel.
\end{enumerate}
\section{Combined ZF-PGP Robust Downlink Precoder}
\subsection{ZF-PGP for Near Optimal Independent Data Stream Distribution}
For this section, we assume without loss of generality that each user's UE comprises two receiving antennas and that the number of VCMBs available to a generic group $g$ ($1\leq g \leq G$) is an even number. These assumptions are made in order to facilitate the description of the concept in conjunction with PGP-WG \cite{TE_TWC}.
From (\ref{full_jsdm}), the receiving equation on the downlink of group $g$ is
\begin{equation}
{\mathbf y}_{d,g} = {\mathbf H}_{g,v}^H {\mathbf P}_g {\mathbf x}_g + {\mathbf n}_g,
\end{equation}
${\mathbf H}_{g,v}^H$ is the VCM group's downlink matrix of size $N_{d,g}\times |{\cal S}_g|$, ${\mathbf y}_{d,g}$ is the group's size $N_{d,g}$ reception vector, and ${\mathbf n}_g$ is the corresponding noise.
Based on the ZF precoding theory, assuming that $|{\cal S}_g| > N_{d,g}$, the ZF precoder is given by
\begin{equation}
{\mathbf P}_{g,ZF} = {\mathbf H}_{g,v}\left({\mathbf H}_{g,v}^H{\mathbf H}_{g,v} \right)^{-1},
\end{equation}
or by applying the Singular Value Decomposition (SVD)\footnote{As a convention, we always employ the SVD with singular values ranked from maximum to minimum going down the main diagonal, and thus the singular vectors follow the same ranking. We call this the natural SVD in the sequel.} of ${\mathbf H}_{g,v}={\mathbf U}_{g,v}{\boldsymbol \Sigma}_{g,v}{\mathbf V}_{g,v}^H$, where ${\mathbf U}_{g,v},~{\boldsymbol \Sigma}_{g,v},~{\mathbf V}_{g,v}$ are the matrices of left singular vectors, the singular values, and the right singular vectors of ${\mathbf H}_{g,v}$, respectively, we get
\begin{equation}
{\mathbf P}_{g,ZF} = {{\mathbf U}}_{g,v}\left[ \begin{array}{c}
    {\tilde{\boldsymbol \Sigma}}_{g,v}^{-1}\\
    {\bf 0}
    \end{array}
     \right]
     {\mathbf V}_{g,v}=\tilde{{\mathbf U}}_{g,v}{\tilde{\boldsymbol \Sigma}}_{g,v}^{-1}{\mathbf V}_{g,v},
\end{equation}
where ${\tilde{\boldsymbol \Sigma}}_{g,v},~{\tilde{\mathbf U}}_{g,v}$ represent the part of non-zero singular values of ${\boldsymbol \Sigma}_{g,v}$, and its corresponding columns of ${\mathbf U}_{g,v}$, respectively. However, this type of ZF precoder will be employing much more power than the original system without precoding. In other words, the transmitted power required in ZF precoding is $\mathrm{ tr}({\mathbf P}_{g,ZF}{\mathbf P}_{g,ZF}^H) = \sum_{i=1}^{N_{d,g}} \frac{1}{s_{g,v,i}^2}\neq N_{d,g}$, where ${s_{g,v,i}}$ is the $i$th non-zero singular value of ${\mathbf H}_{g,v}$. We thus need to normalize the precoder by $\gamma =\sqrt{\frac{N_{d,g}}{\mathrm{tr}(({\mathbf H}_{g,v}^H{\mathbf H}_{g,v})^{-1})}}= \sqrt{\frac{N_{d,g}}{\sum_{i=1}^{N_{d,g}} \frac{1}{s_{g,v,i}^2}}}$, resulting in the following receiving vector
\begin{equation}
{\mathbf y}_{d,g} =  \gamma{\mathbf x}_g + {\mathbf n}_g,
\end{equation}
which shows an $\mathrm{SNR}$ loss, due to the effect of smaller singular values. However, one can see an advantage of ZF precoding due to the creation of $N_{d,g}$ independent streams in the above equation. In this paper, besides comparing the new ZF-PGP precoder to the ZF one, we will also employ the following additional benchmark precoders: a) The RZF precoder (RZFP) \cite{Lee}, b) The Serial Sub-group Combining precoder (SSCP) which normalizes the overall power of the ZF precoder (ZFP) every two consecutive singular values, and c) Parallel Sub-group Combining precoder (PSCP) which normalizes the ZFP power by grouping pairs of extreme singular values. In the following development of the ZF-PGP precoder, we assume that the BS knows the channels perfectly. However, this perfect CSI at the BS assumption is made to allow for the highest ZF-type precoding performance, because when this is not true, the ZF-type precoders considered herein will suffer from cross-interference that will result in significant loss in performance. For the ZF-PGP though, since we have the downlink estimation issue explained below, we can get a sense of its overall performance even with estimated, erroneous CSI, by increasing the level of errors on the downlink alone. This is done in the next section.

The PGP-WG precoder can be determined as follows \cite{TE_TWC}. In group $g$ there are $|{\cal S}_g|$ VCMBs. We assume an even number of VCMBs without loss of generality. For ZF precoding, we need to assume that $N_{d,g}=2K_g \leq |{\cal S}_g|$ (in the opposite case we need to apply CFSDM to accommodate additional downlink users). PGP-WG develops an $N_{d,g}\times N_{d,g}$ near-optimal precoder as follows. It determines the downlink precoder by first employing SVD ${\mathbf P} = {\mathbf U}_{PGP}{\boldsymbol \Sigma}_{PGP}{\mathbf V}_{PGP}^H$. The left singular vector matrix,  ${\mathbf U}_{PGP}$, is determined from the $N_{d,g}$ first (largest) right eigenvectors of the group's downlink channel right singular vectors \cite{Xiao}, i.e.,
\begin{equation}
{\mathbf U}_{PGP} = {\mathbf U}_{g,v,N_{d,g}},
\end{equation}
where the subscript $N_{d,g}$ emphasizes the selection of the singular vectors corresponding to the largest (non-zero) singular vectors.
In PGP-WG, the precoder right singular vector matrix has a block diagonal structure, where each block has size $2\times2$. Thus, the input and output antennas are partitioned in $K_g$ subgroups each subgroup comprising two output antennas and two input symbols. Since different users need to receive different data symbols, we assign same user antennas to one user subgroup. The input symbols are grouped based on the non-zero group's channel singular values, with the optimal way being forming the singular value groups in the PGP by combining the most
distant (in value) singular values, then the best performance is achieved by the PGP \cite{TE_TWC}. Finally, in PGP,
$${\mathbf V}_{PGP}=\mathrm{diag}[{\mathbf V}_1 \cdots {\mathbf V}_{N_{d,g}}]=$$
$$
={\scriptstyle
 \left[ \begin{array}{c c c c c c } {\mathbf V}_{P_1} & {\mathbf 0} & {\mathbf 0}& \cdots & {\mathbf 0} &{\mathbf 0}\\
 {\mathbf 0} & {\mathbf V}_{P_2} & {\mathbf 0} & \cdots  & {\mathbf 0}&  {\mathbf 0} \\
 \vdots & \vdots & \vdots & \ddots & \vdots & \vdots\\
{\mathbf 0} & {\mathbf 0}&{\mathbf 0}&  \cdots & {\mathbf 0}  & {\mathbf V}_{P_{N_{d,g}}} \end{array} \right]
}
$$
is a unitary matrix, and the diagonal matrix
$${\boldsymbol \Sigma}_{PGP} = \mathrm{diag}[{\boldsymbol \Sigma}_1 \cdots {\boldsymbol \Sigma}_{N_{d,g}}]$$
satisfies $\sum_{i=1}^{N_{d,g}} \mathrm {tr}({\boldsymbol \Sigma}_{P_i}^2) = N_{d,g}$. Then, the PGP-WG approach performs a sequence of $K_g$, size $2\times2$ globally optimal precoding determinations, i.e., it solves the following $N_{d,g}$ optimization sub-problems, one for each partition $i$ subgroup ($i=1,2,\cdots,N_{d,g}$),
\begin{eqnarray}
\begin{aligned}
& \underset{{\mathbf P}_{i}}{\text{maximize}}
& & I({\mathbf x}_{s_i};{\mathbf y}_{s_i})\\
& \text{subject to}
& &  \mathrm {tr}({\boldsymbol \Sigma}_{P_i}^2 ) = 2,\\ \label{eq_PGP}
\end{aligned}
\end{eqnarray}
with ${\mathbf P}_i = {\boldsymbol \Sigma}_{P_i}{\mathbf V}_{P_i}$. Thus, PGP-WG needs to run $N_{d,g}$ globally optimal precoders of size $2\times2$, which was done very efficiently for QAM type modulations in the past \cite{TK_EA_QAM}. In addition, PGP-WG has been shown in \cite{LozanoA} to be near-optimal \cite{LozanoA}.
 \setcounter{theorem}{0}
 \begin{theorem}
    For the ZF-PGP precoder and under the JSDM-FA model, the optimal PGP-WG precoder is determined by re-weighting the optimal PGP powers of the VCM channel, due to the presence of the ZF part, while the right PGP-WG singular vectors remain the same as the original PGP-WG solution, i.e., the ones determined without considering a ZF precoder part.
    \end{theorem}
\begin{proof}
The downlink received vector equation for group $g$ ($1\leq g\leq G$) is
\begin{equation}
{\mathbf y}_{d,g}={\mathbf H}_{d,g}^H{\mathbf P}_{g,ZF}{\mathbf P}_{g,PGP}{\mathbf x}_g+{\mathbf n}_g, \label{eq_combo}
\end{equation}
where ${\mathbf P}_{g,ZF}$, ${\mathbf P}_{g,PGP}$ represent the ZF and PGP precoders, respectively. Here, the ZF part is employed in order to equalize the MIMO channel. It is well-known that a mutual information maximizing linear precoder only depends on the singular values of the channel \cite{Xiao}.
    We use the fact that the ZFP part satisfies ${\mathbf H}_{d,g}^H{\mathbf P}_{g,ZF}={\mathbf I}_{N_{d,g}}$. Then, the receiving equation for all antennas after introducing the PGP part can be written as
    \begin{eqnarray}
    \begin{aligned}
{\mathbf y}_{d,g}&=&{\mathbf P}_{g,PGP}{\mathbf x}_g+{\mathbf n}_g. \label{eq_PGP_MOD}
\end{aligned}
\end{eqnarray}
    Using SVD the expressions for ${\mathbf P}_{ZF},~{\mathbf P}_{g,PGP}$ can be written as ${\mathbf P}_{g,ZF} = {\mathbf U}_{g,v}\left[ \begin{array}{c}
    {\tilde{\boldsymbol \Sigma}}_{g,v}^{-1}\\
    {\bf 0}
    \end{array}
     \right]
     {\mathbf V}_{g,v}{\tilde{\boldsymbol \Sigma}}_{g,v},$ and ${\mathbf P}_{g,PGP} = {\tilde{\boldsymbol \Sigma}}_{g,PGP}{\mathbf V}_{g,PGP}^H$, respectively. Thus,
     \begin{equation}
     {\mathbf P}_{g,ZF}{\mathbf P}_{g,PGP} = {\mathbf U}_{g,v} \left[ \begin{array}{c}
   {\tilde{\boldsymbol \Sigma}}_{g,PGP}\\
    {\bf 0}
    \end{array}
     \right]
     {\mathbf V}_{g,PGP}^H={\mathbf U}_{g,v,N_{d,g}}{\tilde{\boldsymbol \Sigma}}_{g,PGP}{\mathbf V}_{g,PGP}^H,
     \end{equation}
     where in order to simplify notation, we defined ${\mathbf U}_{g,v,N_{d,g}}$ to represent the first $N_{d,g}$ highest singular vectors of ${\mathbf U}_{g,v}$.
     Note that then the power constraint for the overall precoder that comprises both the ZF and PGP parts becomes $\mathrm {tr}({\mathbf P}_{g,ZF}{\mathbf P}_{g,PGP}{\mathbf P}_{g,PGP}^H{\mathbf P}_{g,ZF}^H)=N_{d,g}$, or equivalently, after using properties of the matrix trace, we can write for the overall precoder power constraint
     \begin{equation}
     \mathrm {tr}\left({\mathbf V}_{g,v}^H{\tilde{\boldsymbol \Sigma}}_{g,v}^{-2}{\mathbf V}_{g,v}{\tilde{\boldsymbol \Sigma}}_{g,PGP}^2\right)=N_{d,g} \label{eq_power}.
     \end{equation}
    Let us define ${\mathbf M}\doteq {\mathbf V}_{g,v}^H{\tilde{\boldsymbol \Sigma}}_{g,v}^{-2}{\mathbf V}_{g,v}$, then equivalently
    \begin{equation}
     \mathrm {tr}\left({\mathbf M}{\tilde{\boldsymbol \Sigma}}_{g,PGP}^2\right)=N_{d,g} \label{eq_power2}.
     \end{equation}
     We see that due to the ZF precoder part, the overall precoding power constraint becomes involved. In addition, due to the presence of small singular values in the channel, especially in correlated channels, as the ones present here, there will be a significant power loss in the power constraint described by (\ref{eq_power}). Now, since ${\tilde{\boldsymbol \Sigma}}_{g,PGP}$ is a diagonal matrix, we can write (\ref{eq_power2}) as
     \begin{equation}
     \sum_{m=1}^{N_{d,g}} {\mathbf M}_{m,m} s_{PGP,m}^2 = N_{d,g}. \label{eq_power3}
     \end{equation}
     On the other hand, the performance of a mutual information maximizing precoder is determined by ${\tilde{\boldsymbol \Sigma}}_{g,PGP}~\text{and}~{\mathbf V}_{g,PGP}$ according to (\ref{eq_PGP_MOD}). Incorporating the new re-weighted constraints of (\ref{eq_power3}) into the PGP-WG frameork presented above is feasible, but it would result in additional complexity. However, using $w_m \doteq
     \sqrt{|{\mathbf M}_{m,m}|} $, for $m=1,2,\cdots N_{d,g}$, which are calculated in closed form in Appendix A, a simplification is possible by casting the power constraint on the PGP part as an ordinary one. By defining ${\boldsymbol \Sigma}_{g,v,eff} = \mathrm{diag}[s_{v,g,eff,1}\cdots s_{g,v,eff,N_{d,g}}]$, ${\boldsymbol \Sigma}_{g,PGP,eff} = \mathrm{diag}[s_{g,PGP,eff,1}\cdots s_{g,PGP,eff,N_{d,g}}]$, where $s_{g,v,eff,m}\doteq \frac{1}{w_m}$, $s_{g,PGP,eff,m}\doteq {s_{PGP,m}}{w_m}$, respectively, for $m=1,2,\cdots N_{d,g}$, we have the equivalent model fitting within the original PGP-WG type power constraint as follows
     \begin{equation}
     {\mathbf y}_{d,g}={\tilde{\boldsymbol \Sigma}}_{g,v,eff}{\mathbf P}_{g,PGP,eff}{\mathbf x}_g+{\mathbf n}_g={\tilde{\boldsymbol \Sigma}}_{g,v,eff}{\boldsymbol \Sigma}_{g,PGP,eff}{\mathbf V}_{g,PGP}^H{\mathbf x}_g+{\mathbf n}_g, \label{ch_eff}
     \end{equation}
     where the optimal precoder needs to solve
    \begin{eqnarray}
\begin{aligned}
& \underset{{\boldsymbol \Sigma}_{g,PGP,eff},{\mathbf V}_{g,PGP} }{\text{maximize}}
& & I({\mathbf x_d};{\mathbf y_d})\\
& \text{subject to}
& &  \mathrm {tr}({\boldsymbol \Sigma}_{g,PGP,eff}^2) = N_{d,g}. \\ \label{eq_MODIFIED}
\end{aligned}
\end{eqnarray}
   Thus, the optimization is within the framework of PGP-WG type of problems. Then, one can readily apply the original PGP-WG algorithm to solve    (\ref{eq_MODIFIED}).
   \end{proof}
   Based on the previous theorem proof, it is straightforward to prove the following corollary.
    \setcounter{theorem}{0}
   \begin{corollary}
   The ZF-PGP precoder is equivalently created by using an effective channel singular value matrix given by ${\boldsymbol \Sigma}_{g,v,eff}$, followed by an ordinary PGP-WG precoder.
   \end{corollary}
   This corollary demonstrates that the actual ZF-PGP precoder is determined in a straightforward manner by using an effective channel ${\boldsymbol \Sigma}_{g,v,eff}$ and then by determining its corresponding WG-PGP. The diagonal effective channel singular value matrix is determined by the inverses of the $w_m$, $m=1,2,\cdots,N_{d,g}$, which are given by (see Appendix A) $w_m=\sqrt{|{\mathbf M}_{m,m}|}=\sqrt{\sum_{m'=1}^{N_{d,g}}\frac{|({\mathbf V}_{g,v})_{m,m'}|^2}{s_{g,v,m'}^2}}$.

Note that although a ZF-PGP precoder can be developed in a quite straightforward way based on Theorem 1 and Corollary 1, there is one essential question remaining on its performance. Regarding the performance advantages of ZF-PGP, we need to stress the following facts. For uncorrelated MIMO channels, i.e., when full independence exists between the ${\mathbf H}_{g,v}$ entries and with the condition number of the channel being relatively low \cite{Tse} , then a ZF-PGP precoder does not perform well, while a ZF one performs extremely well. The original pioneering works of \cite{Swindle}, \cite{Shamai}, and others have addressed independent Gaussian channels and capitalized on the very desirable property of the ZF precoder to offer separate data streams to each user. When a MIMO channel presents correlations, for example by using a channel model for ULA herein and with a high condition number, due to massive MIMO, ZF and RZF gains evaporate rapidly and a PGP-WG precoder performs much better as demonstrated in the next section. In addition, when the number of UEs becomes close to the number of available VCMBs in the group $|{\cal S}_g|$, although a PGP-WG precoder achieves excellent performance \cite{TK_EA_TCOM2}, it requires a joint decoding at the receiving UEs, and thus it becomes impractical. Under these conditions, by applying CFSDM in conjunction with a ZF-PGP outer MU-MIMO precoder part, one can achieve excellent performance with simplified receiver operation due to the independent data streams created by the inner ZF MU-precoder part. This is demonstrated in the next section with numerical results.

\subsection{Mismatched Downlink Decoding}
We focus on errors on the downlink of the system with each UE being sent pilot data using the optimized ZF-PGP precoder the BS has already determined, so that the UE can estimate the employed precoder, then use the precoder estimate in order to decode the received information data. This is due to the fact that the UE needs to know the optimal precoder determined by the BS in order to decode the data. This estimation process contains errors. In order to assess the system's robustness to this kind of errors, we employ a more general method used extensively in the literature \cite{Caire_Kam, Schoeber_THP} to model the estimation errors. The applied technique helps decouple the employed estimation technique from its corresponding performance, as explained in Section IV. Based on the imperfect channel estimate, the subgroup receiver determines the precoder and employs it at the decoding process. Due to the estimation errors, this results in a mismatch in the UE decoding process. We calculate the corresponding mismatched mutual information with the downlink receiver operating under mismatch, due to the erroneous estimate and assess the corresponding performance loss. We present our methodology toward that end below.

Let us consider the downlink mismatched precoded information decoding under ZF-PGP. Let us denote by ${\mathbf P},~ {\hat {\mathbf P}},~p_{{\mathbf P},{\mathbf y}}({\mathbf y}|{\mathbf x}),~p_{{\hat{\mathbf P}},{\mathbf y}}({\mathbf y}|{\mathbf x}),$ the optimal precoder employed on the downlink, the estimation-based precoder, the pdf of the actual downlink PGP-WG subgroup channel, and the pdf of the estimated downlink PGP-WG subgroup channel, respectively. Then, the mismatched mutual information, $I_{msm}({\mathbf x};{\mathbf y})$, on the subgroup downlink under the presented mismatch can be found employing the results in \cite{Fischer} and using\footnote{The expression used for mutual information under mismatch gives accurate results only if the support of the estimated channel pdf is a superset of the support of the actual channel pdf \cite{Fischer}. If this condition is not satisfied, then negative or very high values for the mismatched mutual information could result which in reality means that the achieved mutual information is zero.}
\begin{equation}
\begin{split}
I_{msm}({\mathbf x};{\mathbf y})&= -\frac{1}{M^{N_t}}\sum_{{\mathbf x}_m}{\mathbb E}_{{\mathbf P},{\mathbf y}|{\mathbf x}_m}\left\{ \log \left(
  \sum_{{\mathbf x}_k}\frac{1}{M^{N_t}} \right. \right.\\
  & \left. \left. \times \frac{1}{\pi^{N_r} {\sigma}^{2N_r}}\exp (-\frac{1}{\sigma^2}||{\mathbf y}-{\hat {\mathbf P}}{\mathbf x}_k||^2)\right)\right\}\\
  &+\frac{1}{M^{N_t}}\sum_{{\mathbf x}_m}{\mathbb E}_{{\mathbf P},{\mathbf y}|{\mathbf x}_m}\left\{\log\left(\frac{1}{\pi^{N_r} {\sigma}^{2N_r}}\exp (-\frac{1}{\sigma^2}||{\mathbf y}-{\hat {\mathbf P}}{\mathbf x}_k||^2)\right)\right\},
\end{split}
\end{equation}
where ${\mathbf P}={\tilde{\boldsymbol \Sigma}}_{g,v,eff}{\boldsymbol \Sigma}_{g,PGP,eff}{\mathbf V}_{g,PGP}$, ${\hat{\mathbf P}}={\hat{\tilde{\boldsymbol \Sigma}}}_{g,v,eff}{\hat{\boldsymbol \Sigma}}_{g,PGP,eff}{\hat{\mathbf V}}_{g,PGP}$, respectively, and where the metrics within the argument of the $\log$ function reflect the fact that estimated channels are employed and in our case $N_r=2$ receiving antennas are employed by each subgroup.
Realizing that under ${{\mathbf P},{\mathbf y}|{\mathbf x}_m}$, ${\mathbf y}= {\mathbf P}{\mathbf x}_m+{\mathbf n},$ and after some straightforward simplifications, we can easily find that
\begin{equation}
\begin{split}
I_{msm}({\mathbf x};{\mathbf y})&= N_t\log(M)-{N_r} -\frac{1}{M^{N_t}}\sum_{{\mathbf x}_m}{\mathbb E}_{\mathbf n}\left\{ \log \left(
  \sum_{{\mathbf x}_k} \right. \right.\\
  & \left. \left. \times \exp (-\frac{1}{\sigma^2}||{\mathbf n}+{\mathbf P}{\mathbf x}_m-{\hat {\mathbf P}}{\mathbf x}_k||^2)\right)\right\}-
  \frac{1}{\sigma^2}||{\mathbf P}-{\hat {\mathbf P}}||_F^2, \label{eq_msm}
\end{split}
\end{equation}
where $||{\mathbf P}-{\hat {\mathbf P}}||_F^2 =\mathrm {tr}\left( ({\mathbf P}-{\hat {\mathbf P}})({\mathbf P}-{\hat {\mathbf P}})^H\right)$ is the Frobenius norm square of ${\mathbf P}-{\hat {\mathbf P}}$.
The third term in the above equation can be very accurately approximated by employing the Gauss-Hermite approximation in a fashion similar to \cite{TK_EA_QAM}, with a simple modification to the expression presented there, i.e., substituting ${\mathbf P}{\mathbf x}_m-{\hat {\mathbf P}}{\mathbf x}_k$ for ${\mathbf P}({\mathbf x}_m-{\mathbf x}_k)$. The details are omitted here due to lack of space. The final expression upon employing the Gauss-Hermite approximation is presented in Appendix B.

We observe that due to the channel estimation induced mismatch, the mutual information expression experiences two impacts: First, it has a modified term in the summation of the expectations over ${\mathbf n}$ that reflect the mismatch between the actual channel and its estimation, and second, it contains a Frobenius norm square of the MSE of the estimation, normalized by $\mathrm {SNR}_{s,d}$. This second term shows that there is a limit of how well an increase of $\mathrm {SNR}_{s,d}$ will compensate for estimation errors.

For ZF-PGP to be efficient, very accurate CSI is assumed available at the BS for the reasons explained above. However, in 5G due to the concept of Cloud-RAN (C-RAN) there are multiple reasons that open the possibility of additional error sources to be present. For example, although an accurate channel estimate needs to be present at the BS, the ZF-PGP optimal precoder determination might take place at the C-RAN and employing cheap, inaccurate processors, or the additional propagation delay between the C-RAN and the BS might result in a suboptimal PGP-WG precoder. In order to assess the system's robustness to this kind of error, we employ a more general method used extensively in the literature to model the estimation errors  \cite{Caire_Kam, Schoeber_THP}. The applied technique helps decouple the employed estimation technique from its corresponding performance, as explained in Section IV. Based on the imperfect channel estimate, the subgroup receiver determines the precoder and employs it at the decoding process. Due to the estimation errors, this results in a mismatch and then we calculate the corresponding mismatched mutual information with the downlink receiver operating under mismatch, due to the erroneous estimate and assess the corresponding performance loss.

\section{Numerical Results}
In this section, we present numerical results regarding the precoders presented above and their comparison with other widely used precoders which we use as benchmarks. More explicitly, we use the ZFP, SSCP, PSCP, and the RZFP as benchmarks. In addition, we compare the performance of the ZF-PGP with PGP-WG \cite{TK_EA_TCOM2}. In all cases, there are group channels created based on the JSDM-FA model \cite{TK_EA_TCOM2} which we employ to get the desired results using the methodology presented in \cite{TK_EA_QAM}. We only show results for QAM with $M=16$, however the same approach we used in \cite{TK_EA_QAM,TK_EA_TCOM2} for higher modulation size of $M=64$ can be used here as well. The model used for CSI errors on the downlink is as follows. In all cases we use the Virtual Additional Antennas Concept (VAAC) that we have used in previous work \cite{TK_EA_QAM,TK_EA_TCOM2}. When results with mismatched decoding at the receivers are used, the following model applies. The BS determines the near-optimal precoder for this subgroup and sends this to subgroup $s_g$ ($1
\leq s_g \leq K_g)$, $1 \leq g_s \leq \frac{N_{d,g}}{2}$ of the JSDM-FA model using the ZF-PGP approach as
\begin{equation}
{\mathbf Y}_{s_g} = {\boldsymbol \Sigma}_{s_g}{\mathbf P}_{s_g,o}{\mathbf X}_{p_{s_g}}+{\mathbf N}_{s_g},
\end{equation}
where ${\mathbf Y}_{g_s},~{\boldsymbol \Sigma}_{g_{v,g_s}}, ~ {\mathbf P}_{g_s,o},~{\mathbf X}_{p_{g_s}},~\text{and}~{\mathbf N}_{g_s}$ represent the received data, the channel singular values for the PGP-WG subgroup, the ZF-PGP determined precoder for the PGP-WG subgroup, the pilots, and the noise, respectively, all for subgroup $g_s$. In order to model the errors of the estimation process we use for the estimate of ${\mathbf P}_{g_s,o}$, ${\hat{\mathbf P}}_{g_s,o}$, the following model \cite{Caire_Kam, Schoeber_THP}:
\begin{equation}
{\hat{\mathbf P}}_{s_g,o} = \sqrt{1-\tau^2}{\mathbf P}_{g_s,o}+\tau{\mathbf N}_{est},
\end{equation}
with ${\mathbf N}_{est}$ representing the random errors being a $2\times2$ matrix of complex, independent, Gaussian random variables of mean 0 and variance 1, and with $\tau$ being an estimation quality parameter in $[0,1]$ with 0 representing ideal estimation and 1 representing fully erroneous estimation. We use $\tau=0.1,~0.2,\text{and}~ 0.3$ in the results in this paper. For cases where $\tau \geq 0.2$ is employed, we can make the argument that due to the higher error level, one can model possible channel estimation errors as well and use the corresponding results as showing robustness to both type of errors (at the BS and UE). The groups presented are developed from a massive MIMO system with $N_u=100$ antennas using the JSDM-FA model presented in \cite{TK_EA_TCOM2}, according to the relevant equations presented in Section II. When no comparison is made with the actual channel GCMI in the results, we employ the symbol $\mathrm{SNR}$ as the resource variable, while when comparisons are made to the actual channel GCMI in the results, we employ as resource the information bit based $\mathrm{SNR}$, denoted as $\mathrm{SNR}_b$.

In Fig. 1 we present results using group $G1$ of \cite{TK_EA_TCOM2} for the achievable mutual information with PGP-WG, ZFP, SSCP, PSCP, and RZF, all with ideal CSI. The resulting downlink channel is a $4\times8$ MIMO channel, i.e., $N_{d,g}=4,~|{\cal S}_g|=8$, using an angular spread of $\pm 4^\circ$. We see what is typical behavior in correlated channels, i.e., PGP-WG offers much higher mutual information than all the others, albeit at the cost of requiring demodulation cooperation among the UE. It thus becomes appealing to consider precoding techniques that aim to bridge the gap between ZFP and PGP-WG type of precoding techniques. This channel presents very high correlation and ZF-PGP cannot offer an efficient solution. Instead, we resort to CFSDM in order to reduce the channel correlation and present
\begin{figure}[h]
\centering
\setcounter{figure}{0}
\includegraphics[height=3.6in,width=5.25in]{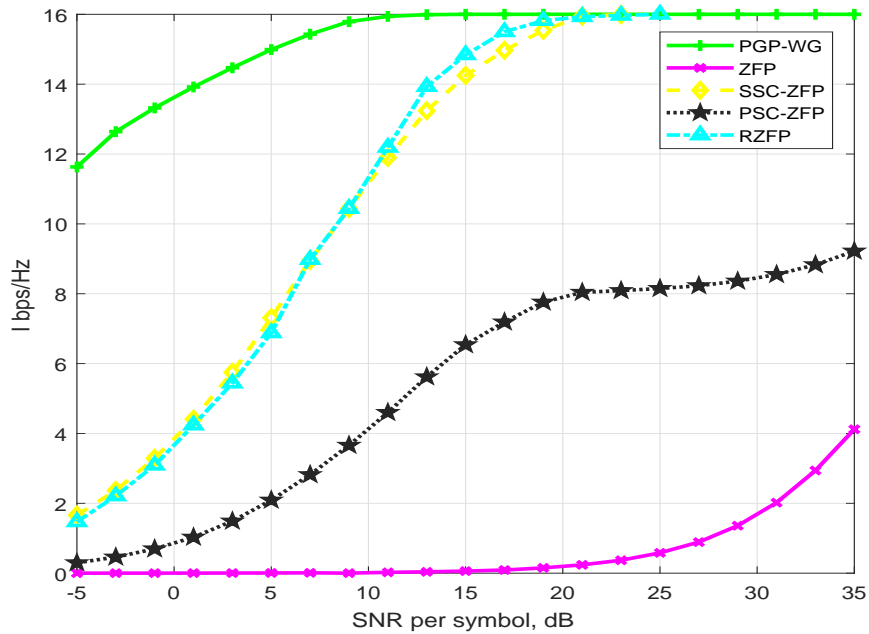}
	\caption{$I({\mathbf x};{\mathbf y})$ results for PGP-WG, ZFP, SSCP, PSCP, and RZFP cases for the channel in $G_1$ of \cite{TK_EA_TCOM2} in conjunction with QAM $M=16$ modulation.}
\end{figure}
our first ZF-PGP results in Fig. 2. Here we present results for the same channel in Fig. 1, but employing two different subcarriers, in a CFSDM fashion, resulting in two CFSDM groups of 1 UE each. We use perfect CSI and employ the Virtual Additional Antenna Concept (VAAC) for the ZF-PGP (explained in Appendix C) with $N_{TV}=2$ additional transmitting antennas added per CFSDM group. We divide the sum of the two CFSDM groups mutual information by two (the number of subcarriers employed in CFSDM) in order to calculate the resulting spectral efficiency. We see that ZF-PGP significantly outperforms ZFP, but due to the exploitation of two subcarriers, the performance of ZF has been improved from the non-CFSDM case dramatically. In addition, the performance of ZF-PGP approaches the performance of PGP-WG at high $\mathrm{SNR}$.
\begin{figure}[h]
\centering
\setcounter{figure}{1}
\includegraphics[height=3.6in,width=5.25in]{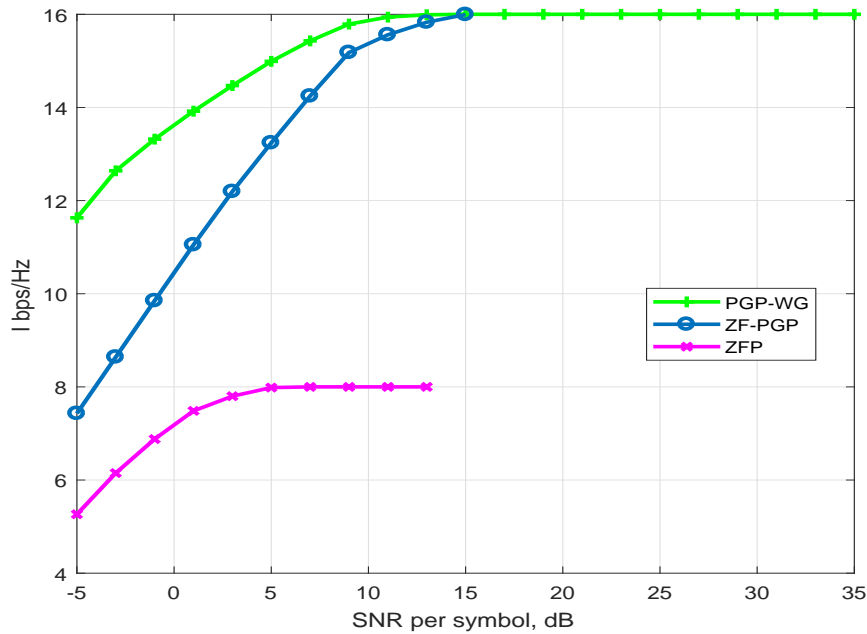}
	\caption{$I({\mathbf x};{\mathbf y})$ results for ZF-PGP, PGP-WG, and ZFP cases for the channel in $G_1$ of \cite{TK_EA_TCOM2} in conjunction with QAM $M=16$ modulation and CFSDM.}
\end{figure}

In Fig. 3 we present results for another small group. Here the group is created based on the JSDM-FA framework \cite{TK_EA_TCOM2} and comprises $N_{d,g}=4,~|{\cal S}_g|=6$, the group was developed with an angular spread of $\pm 10^\circ$. The results presented for ZF-PGP take advantage of the Virtual Additional Antenna Concept (VAAC) that can be very efficiently employed by ZF-PGP and PGP-WG, but it is not possible for the other techniques presented. The example shown uses $N_{TV} = N_{d,g}=4$ added at the transmitter based on VAAC. We observe the excellent performance of ZF-PGP, even with estimation errors. In addition, we see that in high $\mathrm{SNR}$, ZF-PGP follows PGP-WG very closely. The gain of ZF-PGP over ZFP with $\tau=0.1$ level of errors are 42.8\%, and 75\%, at $\mathrm{SNR}_s=5~dB$, and $15~dB$, respectively.

In Fig. 4 we show results for a large group, i.e., $N_{d,g}=30,~|{\cal S}_g|=14$, created with an angular spread of $\pm 4^\circ$ with the JSDM-FA framework. Because $N_{d,g}> |{\cal S}_g|$ in this case, CFSDM is required in order to implement the ZFP type of precoders. We divide the UE in three CFSDM subgroups each comprising 10 antennas. We present results for the third subgroup. The PGP-WG and ZF-PGP systems use $N_{TV}=10$. We observe the excellent performance of ZF-PGP, even with estimation errors. In addition, we see that in high $\mathrm{SNR}$, ZF-PGP follows PGP-WG very closely. The gain of ZF-PGP with $\tau=0.1$ level of errors over ZFP are 42.8\%, and 75\%, at $\mathrm{SNR}_s=5~dB$, and $15~dB$, respectively.
\begin{figure}[h]
\centering
\setcounter{figure}{2}
\includegraphics[height=3.6in,width=5.25in]{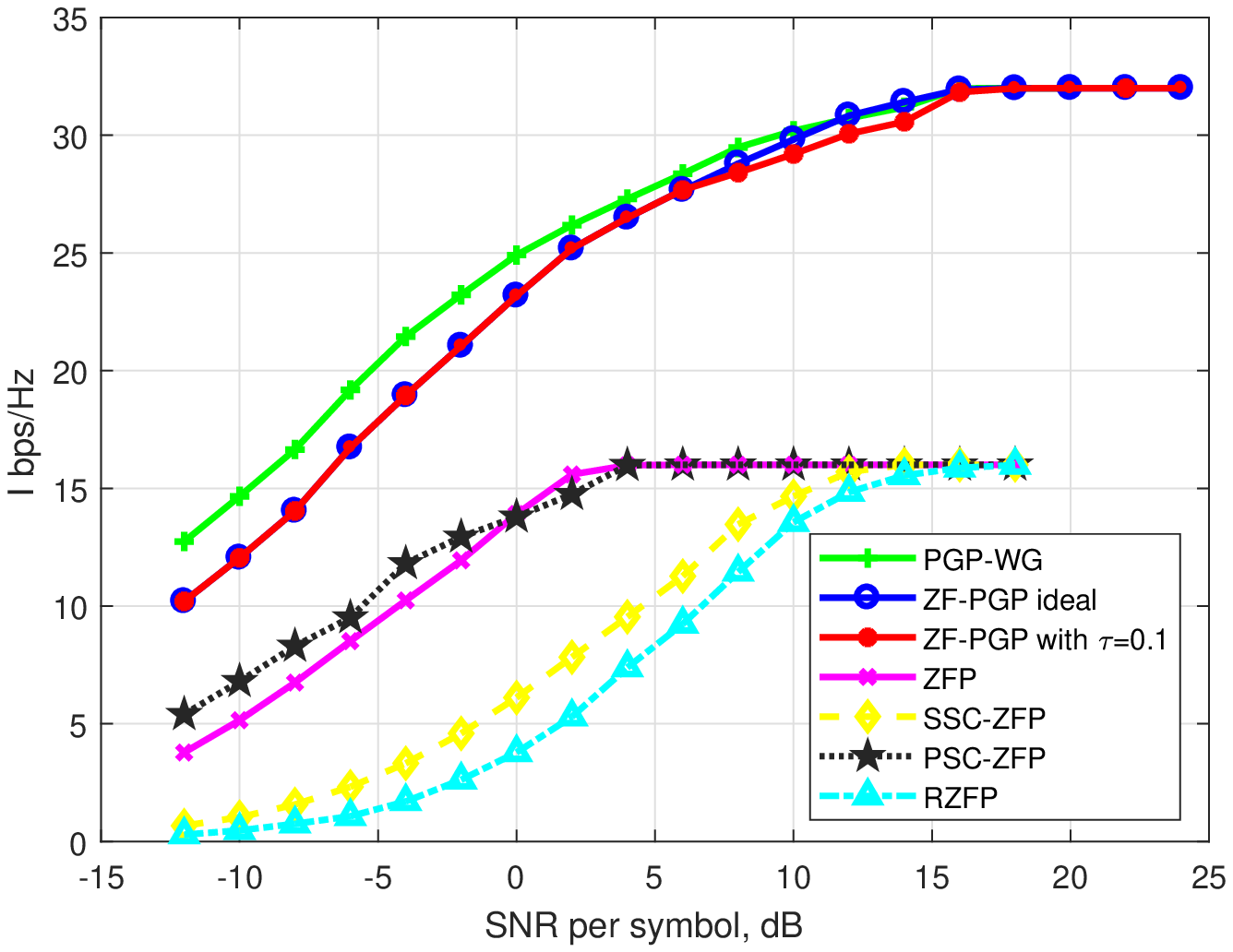}
	\caption{$I({\mathbf x};{\mathbf y})$ results for ZF-PGP (ideal and with channel estimation errors ($\tau=0.1$)), PGP-WG, ZFP, SSCP, PSCP, and RZFP cases for the $4\times6$ group channel in conjunction with QAM $M=16$ modulation.}
\end{figure}
In Fig. 5 we show results for a system with $N_{d,g}=40,~|{\cal S}_g|=12$, created with an angular spread of $\pm 4^\circ$ with the JSDM-FA framework. Here, again CFSDM is needed. Thus, we divide the groups into four equal size subgroups and simulate the first subgroup. The VAAC employs $N_{TV}=10$. The resulting gain of ZF-PGP with $\tau=0.1$ level of errors over ZFP at $\mathrm{SNR}_s=10~dB$ is greater than $50\%$. We observe that the estimation errors have a more profound effect in high $\mathrm{SNR}$ in this case, resulting in a reversing effect as the achievable spectral efficiency peaks in value, then it starts decreasing. This is a typical behavior under the model used, as the estimation errors are independent than the $\mathrm{SNR}$ in the model. This can be the case, for example, if the downlink channel estimation employs a constant estimation $\mathrm{SNR}$, independent of the data transmission $\mathrm{SNR}$.
\begin{figure}[h]
\centering
\setcounter{figure}{3}
\includegraphics[height=3.6in,width=5.25in]{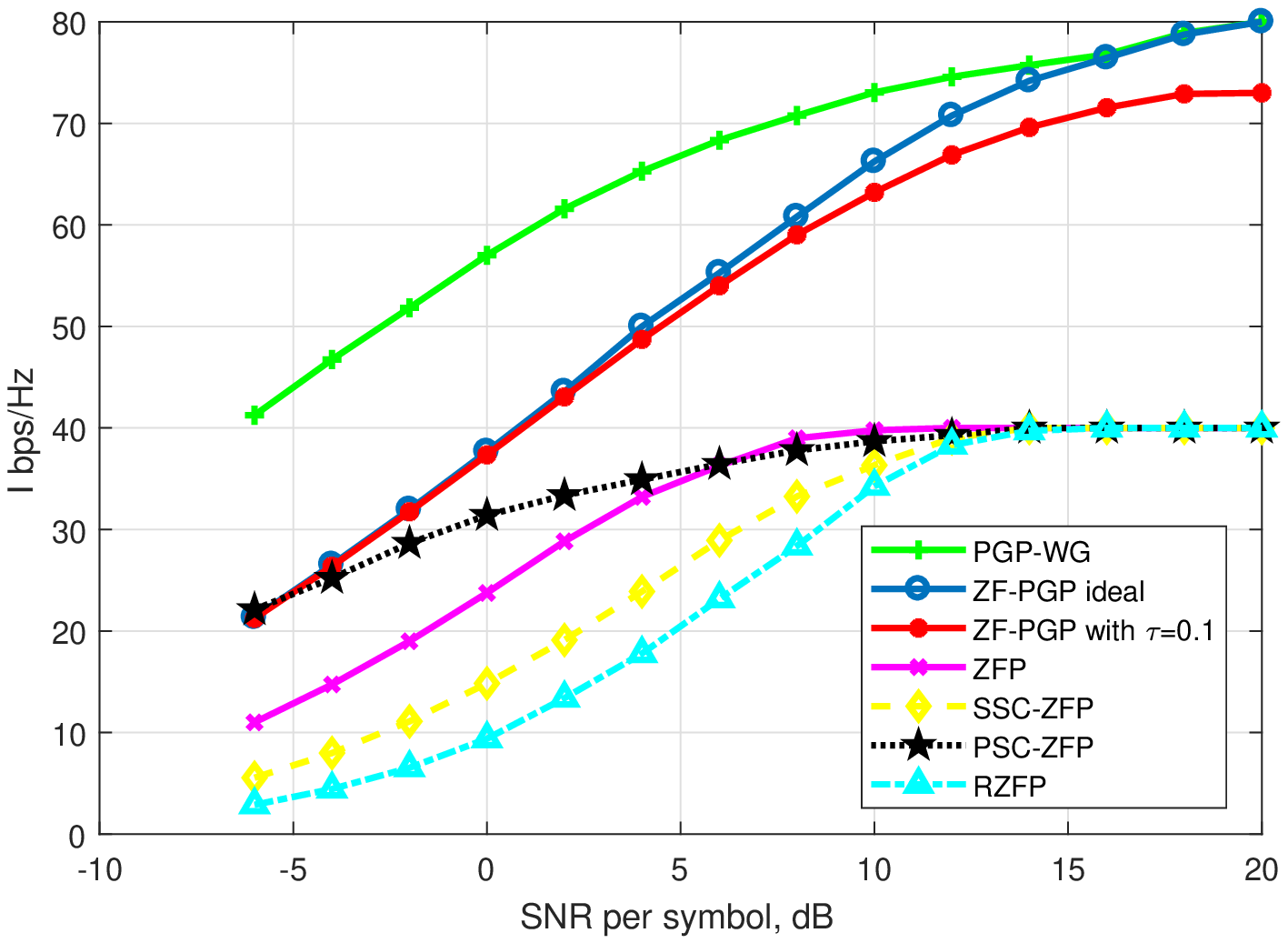}
	\caption{$I({\mathbf x};{\mathbf y})$ results for ZF-PGP (ideal and with channel estimation errors ($\tau=0.1$)), ZFP, SSCP, PSCP, and RZFP cases for the $30\times14$ CFSDM subgroup 3 channel in conjunction with QAM $M=16$ modulation and CFSDM.}
\end{figure}

\begin{figure}[h]
\centering
\setcounter{figure}{4}
\includegraphics[height=3.6in,width=5.25in]{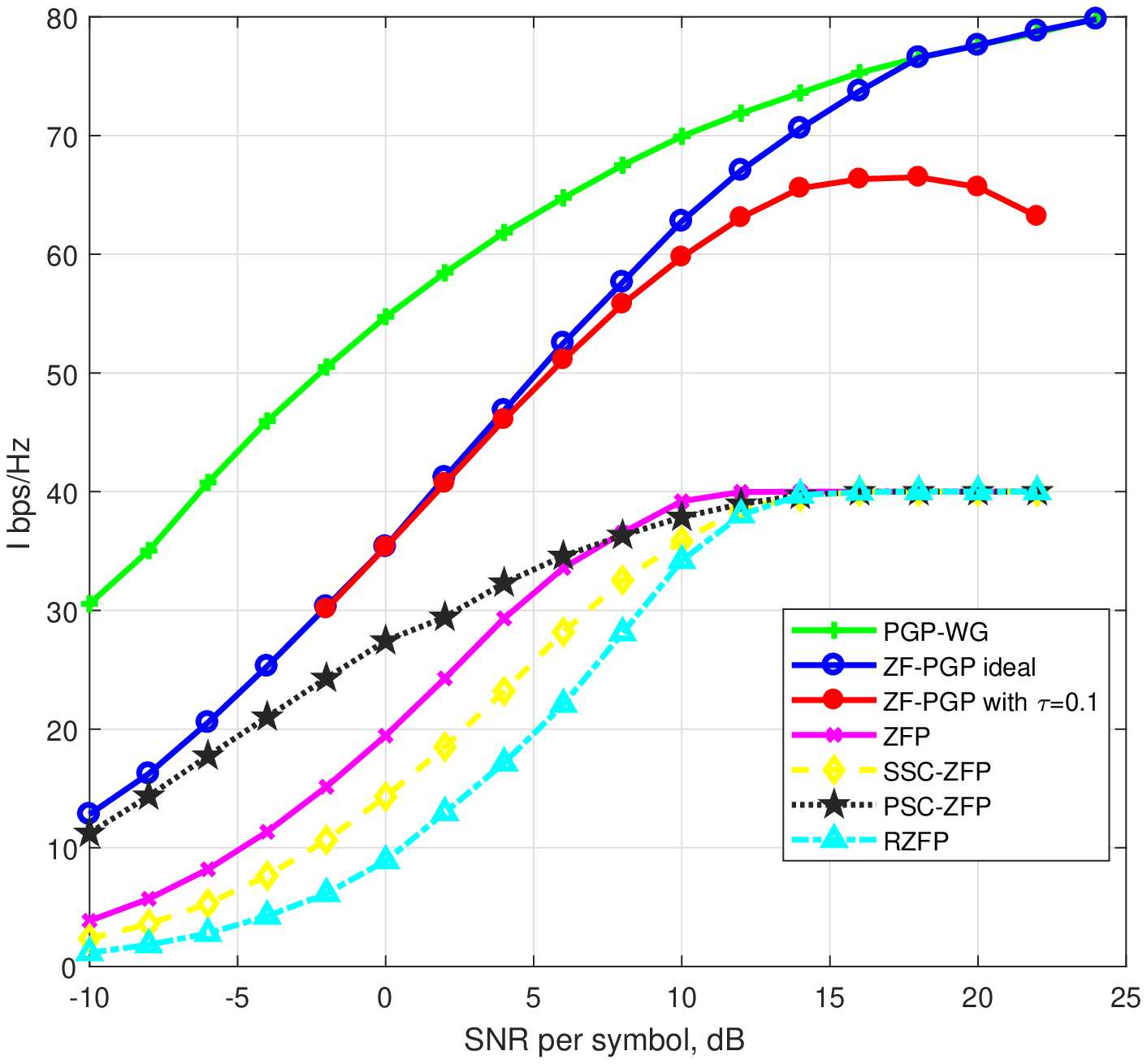}
	\caption{$I({\mathbf x};{\mathbf y})$ results for ZF-PGP (ideal and with channel estimation errors ($\tau=0.1$)), PGP-WG, ZFP, SSCP, PSCP, and RZFP cases for the $40\times12$ group CFSDM subgroup 1 channel and in conjunction with QAM $M=16$ modulation CFSDM.}
\end{figure}
In all cases considered so far, the gains offered by ZF-PGP in low $\mathrm{SNR}$ over ZFP are even higher than the ones in high $\mathrm{SNR}$.

In the rest of the results, we compare the performance of ZF-PGP with the original channel GCMI and ZFP, albeit based on $\mathrm{SNR}_b$. This approach gives a better idea about the power efficiency of different precoding techniques. In Fig. 6, 7, and 8 we present the corresponding results for the groups used in Fig. 3, 4, and 5, respectively.

In Fig. 6 we observe that for all $\mathrm{SNR}_b>4~dB$, ZF-PGP with error level equal to $0.1,~0.2,~\text{and}~0.3$ offers higher spectral efficiency than the original channel GCMI with ideal CSI at the receiving UEs. Due to the increased error levels, we can safely contend that if estimation errors are present on the uplink channel estimation, the ZF-PGP with estimation errors at both the BS and at the UEs will still perform significantly better than the ZFP with ideal CSI, a major conclusion in this paper.
\begin{figure}[h]
\centering
\setcounter{figure}{5}
\includegraphics[height=3.6in,width=5.25in]{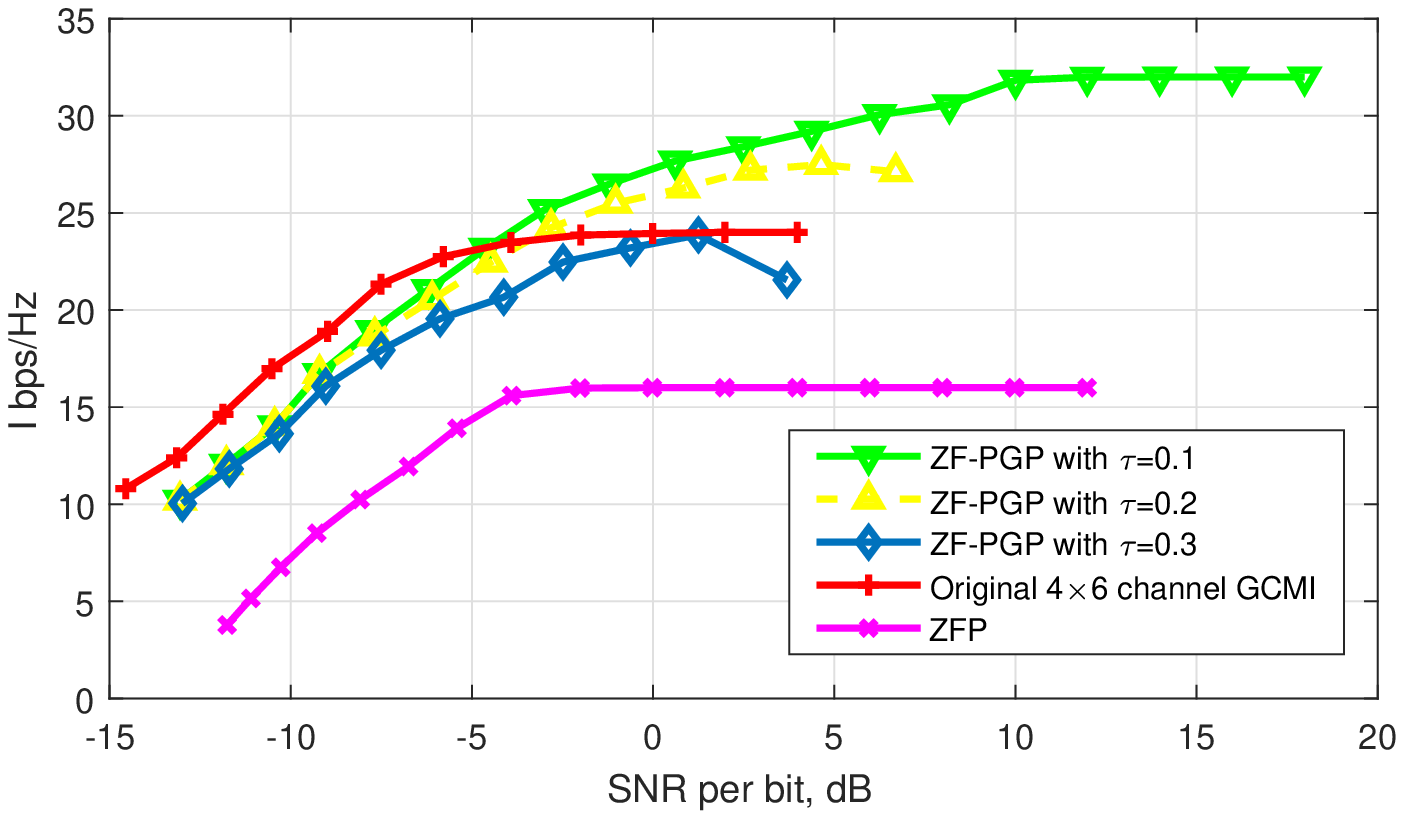}
	\caption{$I({\mathbf x};{\mathbf y})$ results with respect to $\mathrm{SNR}_b$ for ZF-PGP (with channel estimation errors ($\tau=0.1,~0.2,~ \text{and}~ 0.3$)), ZFP, and the original channel GCMI for the $4\times6$ group channel and in conjunction with QAM $M=16$ modulation.}
\end{figure}
In Fig. 7 the same type of behavior is observed, i.e., for $\mathrm{SNR}>3~dB$ a ZF-PGP system with error level $\tau=0.1$ offers higher spectral efficiency than the original channel GCMI with ideal CSI at the UEs.
\begin{figure}[h]
\centering
\setcounter{figure}{6}
\includegraphics[height=3.6in,width=5.25in]{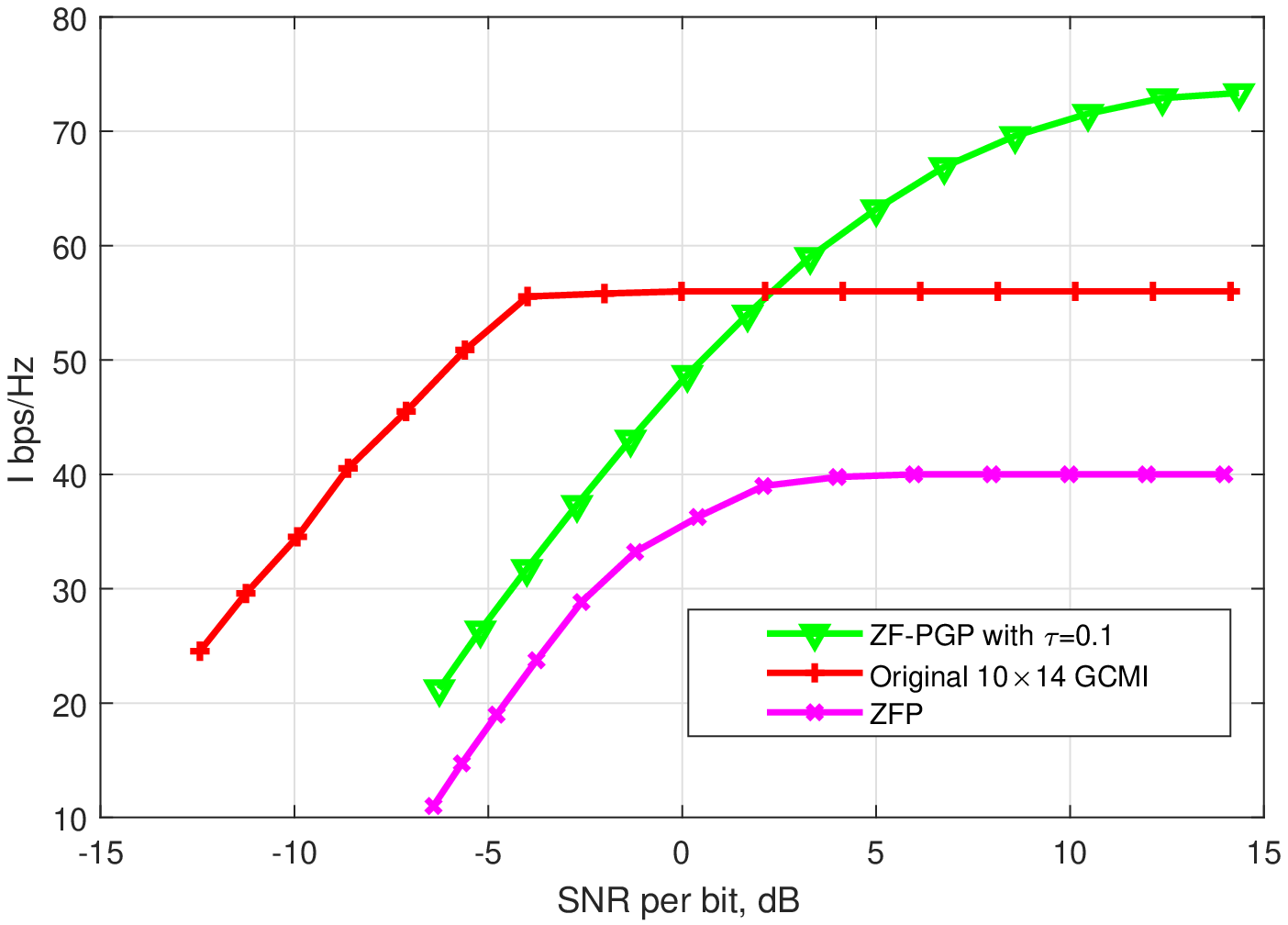}
	\caption{$I({\mathbf x};{\mathbf y})$ results with respect to $\mathrm {SNR}_b$ for ZF-PGP (with channel estimation errors ($\tau=0.1$)), ZFP, and the original channel GCMI for the $30\times14$ group CFSDM subgroup 3 channel and in conjunction with QAM $M=16$ modulation and CFSDM.}
\end{figure}
Finally, in Fig. 8 we see that ZF-PGP with error level $\tau=0.1$ offers higher spectral efficiency than the original channel GCMI with ideal CSI at the UEs. In all cases considered with regards to $\mathrm{SNR}_b$, ZF-PGP with errors is much more power efficient than ZFP as shown by the plots. It is worthwhile stressing that PGP-WG employed in \cite{TK_EA_TCOM2} is always equal or better in spectral efficiency over the original GCMI one with respect to power efficiency ($\mathrm{SNR}_b$).
\begin{figure}[h]
\centering
\setcounter{figure}{7}
\includegraphics[height=3.6in,width=5.25in]{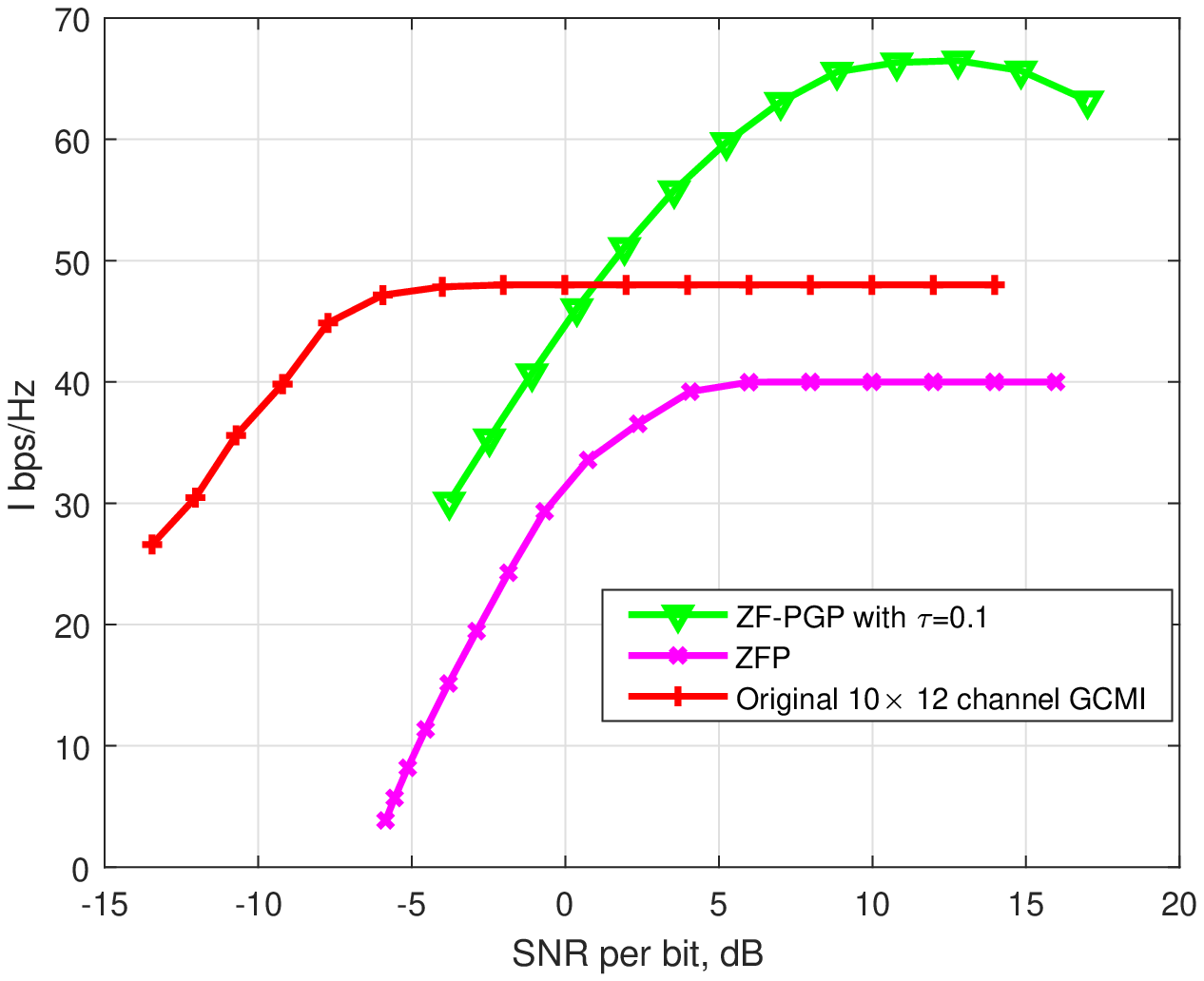}
	\caption{$I({\mathbf x};{\mathbf y})$ results with respect to $\mathrm {SNR}_b$ for ZF-PGP (with channel estimation errors ($\tau=0.1$)), ZFP, and the original channel GCMI for the $40\times12$ group CFSDM subgroup 1 channel and in conjunction with QAM $M=16$ modulation and CFSDM.}
\end{figure}

\section{Conclusions}
In this paper, we propose a new, combined Zero-Forcing Per-Group Precoding (ZF-PGP) method that achieves very high spectral efficiency gains in comparison to ZF precoding techniques, while it simultaneously offers individually separate streams to reach individual UE, i.e., it obliterates the need for coordinated, joint decoding by the group's UEs. We show that for correlated channels and/or when there are more antennas in a group than the number of VCMBs in the group, the near-optimal ZF-PGP is determined from an effective channel singular value matrix which is calculated easily from the singular values and the right singular vectors of the group's uplink virtual channel in the JSDM-FA decomposition. We show that at high SNR ZF-PGP approaches the original PGP-WG near-optimal performance. In addition, ZF-PGP is shown to be robust to estimation errors. The gains of ZF-PGP over ZF type precoders are documented by many examples and shown to be significant, i.e., higher than 50\% in high $\mathrm{SNR}$ even with errors. Finally, ZF-PGP offers higher performance than the original channel GCMI one, although the latter is assumed with perfect CSI at the receiving UEs.
\appendices
\section{Expressions for $w_m$}
From Section III the expression for ${\mathbf M}_{m,m}$ with $m=1,2,\cdots, N_{d,g}$ is
\begin{equation}
{\mathbf M}_{m,m} = ({\mathbf V}_{g,v})_{m,\cdot}{\tilde{\boldsymbol \Sigma}}_{g,v}^{-2}({\mathbf V}_{g,v})_{m,\cdot}^H=\sum_{m'=1}^{N_{d,g}}\frac{|({\mathbf V}_{g,v})_{m,m'}|^2}{s_{g,v,m'}^2}. \label{eq_w}
\end{equation}
Upon taking the square root of (\ref{eq_w}), the desired expression is obtained.

\section{Gauss-Hermite approximation to $I_{msm}$}
In expression (\ref{eq_msm}), we can apply a similar Gauss-Hermite approximation to the one applied in \cite{TK_EA_TCOM2}. Using this, we can approximate (\ref{eq_msm}) as
\begin{equation}
\begin{split}
&I({\mathbf x};{\mathbf y})  \approx N_t\log_2(M) -\frac{N_r}{\log(
2)} -\frac{1}{M^{N_t}}\sum_{k=1}^{M^{N_t}}{\hat f}_k,\\
\label{eq_PAPER}
\end{split}
\end{equation}
where
\begin{equation}
\begin{split}
&{\hat f}_k =  \left(\frac{1}{\pi}\right)^{N_r}\sum_{k_{r1}=1}^{L} \sum_{k_{i1}=1}^{L}\cdots \sum_{k_{rN_r}=1}^{L} \sum_{k_{iN_r}=1}^{L} c(k_{r1})c(k_{i1})\cdots \\
& c(k_{rN_r})c(k_{iN_r})g_k(\sigma n_{{k_{r1}}}, \sigma n_{{k_{i1}}},\cdots,\sigma n_{{k_{rN_r}}}, \sigma n_{{k_{iN_r}}}), \label{eq_f}
\end{split}
\end{equation}
with
\begin{equation}
g_k(\sigma v_{{k_{r1}}}, \sigma v_{{k_{i1}}},\cdots,\sigma v_{{k_{rN_r}}}, \sigma v_{{k_{iN_r}}})
\end{equation}
being the value of the function
\begin{equation}
\log_2\left(\sum_{m}\exp (-\frac{1}{\sigma^2}||{\mathbf n}-{\hat{{\mathbf H}}}_{g,v}^H{\hat{\mathbf P}}{\mathbf x}_k+{\mathbf H}_{g,v}^H{\mathbf P}{\mathbf x}_m||^2)\right) \label{eq_basic}
\end{equation}
evaluated at ${\mathbf n}_e =\sigma{\mathbf v}(\{k_{rv},~k_{iv}\}_{v=1}^{N_r})$,
where for the model in this paper, $N_t=N_r=N_{d,g}$.

\section{VAAC Rationale}
Here the concept of adding virtual antennas, i.e., additional data streams to the same antennas employed by a MIMO system is explained in detail. Without a loss of generality, we consider a MIMO system with equal number of transmitting and receiving antennas, i.e., $N_t=N_r=N$, where $N_t$, and $N_r$ represent the number of transmitting, and receiving antennas, respectively. The channel model under consideration then becomes
\begin{equation}
{\mathbf y}={\mathbf H}{\mathbf x}+{\mathbf n},
\end{equation}
where ${\mathbf y},~{\mathbf H},~{\mathbf x},~\text{and}~{\mathbf n}$ represent the received data, the MIMO channel, the transmitted data, and the AWGN noise, respectively, and where  matrices are of size $N\times N$ and vectors are of size $N\times 1$. The equivalent singular value decomposition based model for the MIMO channel is
\begin{equation}
{\mathbf y}={\mathbf U}_H{\boldsymbol \Sigma}_H{\mathbf V}_H^H{\mathbf x}+{\mathbf n}, \label{eq_MIMO}
\end{equation}
with ${\mathbf U}_H,~{\boldsymbol \Sigma}_H,~{\mathbf V}_H$ representing the size $N\times N$ matrices of left singular vectors, singular values, and right singular vectors, respectively.
Consider adding $N$ virtual antennas of zero singular values, i.e., useless, noise-only channels. This can be added to the previous model as follows

\begin{eqnarray}
\begin{aligned}
\left[\begin{array}{c  } {\mathbf y}\\
{\mathbf y}_a\end{array} \right]&=&
\left[\begin{array}{c c } {\mathbf U}_H & {\mathbf 0}\\
{\mathbf 0} & {\mathbf I}_N\end{array} \right]
\left[\begin{array}{c c } {\boldsymbol  \Sigma}_H & {\mathbf 0}\\
{\mathbf 0} & {\mathbf 0}\end{array} \right]
\left[\begin{array}{c c } {\mathbf V}_H^H & {\mathbf 0}\\
{\mathbf 0} & {\mathbf I}_N\end{array} \right]
\left[\begin{array}{c  } {\mathbf x}\\
{\mathbf x}_a\end{array} \right] +
\left[\begin{array}{c  } {\mathbf n}\\
{\mathbf n}_a\end{array} \right], \label {eq_VAAC}
\end{aligned}
\end{eqnarray}
where the subscript $a$ is used to indicate the $N$ added, fictitious antennas. In the above equation, the vector ${\mathbf x}_a$ represents the $N$ added QAM inputs to the MIMO system. Note that in (\ref{eq_VAAC}) the inputs represented by ${\mathbf x}_a$ cannot be transmitted, due to their corresponding zero singular values (noise-only channel) which result in zero input-output mutual information. However, one can still apply the virtual model of (\ref{eq_VAAC}) with PGP-WG. The PGP-WG algorithm will optimize and assign an amplitude diagonal matrix as per PGP-WG ${\boldsymbol \Sigma}_{P_i}=\mathrm{diag}[\sqrt{ 2}, ~0]$, $i=1,2,\cdots,K_g$ for each sub-group of the PGP-WG, i.e., no power sent to the noise-only antenna. On the other hand, PGP-WG will also determine the optimal unitary precoder matrix to each sub-group in the PGP-WG, thus it will be multiplexing optimally two QAM symbols to each actual transmitting antenna of the original MIMO system in (\ref{eq_MIMO}).
\bibliographystyle{IEEEtran}

\end{document}